\documentclass[11pt,reqno]{amsart}

\usepackage[T1]{fontenc}
\usepackage[utf8]{inputenc}
\usepackage{lmodern}
\usepackage{microtype}
\usepackage{amsmath,amssymb,amsthm,mathtools}
\usepackage{mathrsfs}
\usepackage{enumitem}
\usepackage{booktabs}
\usepackage{array}
\usepackage[hidelinks]{hyperref}

\setlist[itemize]{leftmargin=1.6em,itemsep=0.3em,topsep=0.3em}
\setlist[enumerate]{leftmargin=1.8em,itemsep=0.3em,topsep=0.3em}

\numberwithin{equation}{section}

\newtheorem{theorem}{Theorem}[section]
\newtheorem{proposition}[theorem]{Proposition}
\newtheorem{lemma}[theorem]{Lemma}
\newtheorem{corollary}[theorem]{Corollary}

\theoremstyle{definition}
\newtheorem{definition}[theorem]{Definition}
\newtheorem{example}[theorem]{Example}

\theoremstyle{remark}

\DeclareMathOperator{\ord}{ord}
\DeclareMathOperator{\PV}{PV}

\newcommand{\ii}{\mathrm{i}}

\title[Polynomial rigidity of strong-field magnetic billiards]
{Polynomial rigidity of strong-field magnetic billiards}

\author{Dipesh Bhandari}
\address{Department of Physics, Southern Methodist University, Dallas, Texas 75275, USA}
\email{dbhandari@smu.edu}

\date{\today}

\subjclass[2020]{37J40, 37J35, 14H50}
\keywords{magnetic billiards, polynomial first integrals, Birkhoff rigidity, algebraic curves, normalization, Larmor centers, strong magnetic field}

\begin{document}

\begin{abstract}
Let \(\Omega\subset\mathbb R^2\) be a bounded domain whose boundary \(\gamma\) is a smooth simple closed curve, and let \(r=|B|^{-1}\) be the Larmor radius of a constant magnetic field. Assume
\[
0<r<\frac{r_0(\gamma)}2,
\]
where \(r_0(\gamma)\) is the maximal embedded tubular radius of \(\gamma\). We prove that if the corresponding magnetic billiard admits a nonconstant first integral polynomial in the velocity, then \(\gamma\) is a circle and \(\Omega\) is a disk. No convexity assumption is imposed.

This eliminates the finite set of possible exceptional strong-field values left by the earlier polynomial nonintegrability theory. The proof uses the two highest reflection modes of the polynomial induced on the space of Larmor centers. A boundary index identity shows that the highest coefficient has degree equal to the true Fourier degree and that all of its roots lie inside the table. After complexification and passage to the normalization of the boundary curve, comparison of the two leading modes at infinity excludes simultaneous poles of the coordinate functions and forces all of these roots to coincide. The reflection identity then implies that the boundary has constant distance from their common root and hence is circular.
\end{abstract}

\maketitle

\section{Introduction}

A planar magnetic billiard describes a unit-speed charged particle moving in a bounded domain \(\Omega\subset\mathbb R^2\) under a constant magnetic field perpendicular to the plane. Between collisions with the boundary, the particle follows oriented circular Larmor arcs of radius
\[
r=\frac1{|B|},
\]
and at the boundary its velocity undergoes specular reflection. Magnetic billiards were introduced by Robnik and Berry and have since been studied from dynamical, geometric, algebraic, and semiclassical viewpoints; see, for example, \cite{RobnikBerry1985,BerglundKunz1996,Bialy2012,Tabachnikov2004,BialyMironov2016,BialyMironovShalom2020}.

The round disk is integrable. If the disk is centered at the origin, then
\begin{equation}
h(x,v)=x_1^2+x_2^2+\frac{2}{\beta}(v_1x_2-v_2x_1),
\qquad \beta=|B|,
\label{eq:disk-integral}
\end{equation}
is preserved by both the magnetic flow and the reflection law. This raises a natural rigidity question: can a noncircular magnetic billiard admit a nonconstant first integral polynomial in the velocity?

Polynomial integrability is an exact and deliberately restrictive notion. It is not equivalent to total integrability, but it has the advantage that dynamical invariance becomes a finite collection of algebraic identities. The problem here is therefore not whether a strong magnetic field can produce approximately regular motion, but whether an exact finite-degree algebraic integral can survive for a noncircular boundary.

We work in the strong-field regime
\begin{equation}
0<r<\frac{r_0(\gamma)}2,
\label{eq:strong-field}
\end{equation}
where \(r_0(\gamma)\) is the maximal radius of an embedded tubular neighborhood of the boundary. In this regime the parallel curves needed by the magnetic dynamics remain smooth and embedded. The Larmor-center formulation of Bialy, Mironov, and Shalom turns a polynomial first integral into a polynomial invariant on the annulus of Larmor centers and forces the relevant parallel curves to be nonsingular algebraic ovals \cite{BialyMironovShalom2020}. Their strong-field setup is formulated for a smooth simple closed boundary and does not require convexity.

Earlier algebraic work of Bialy and Mironov imposed strong restrictions on polynomially integrable magnetic billiards and, in particular, excluded noncircular ellipses \cite{BialyMironov2016}. In the strong-field regime, Bialy, Mironov, and Shalom proved that for a fixed noncircular table polynomial integrability is impossible for all but finitely many field magnitudes \cite{BialyMironovShalom2020}. The remaining possibility was therefore a finite set of isolated exceptional strong fields. The main result eliminates that possibility.

\begin{theorem}[Strong-field polynomial rigidity]\label{thm:main}
Let \(\Omega\subset\mathbb R^2\) be a bounded domain whose boundary
\[
\gamma=\partial\Omega
\]
is a \(C^\infty\) simple closed curve. Let \(r=|B|^{-1}\) be the Larmor radius of a constant nonzero magnetic field and assume \eqref{eq:strong-field}.

Assume that the magnetic billiard admits a nonconstant first integral
\begin{equation}
\Phi(x,v)=\sum_{0\le k+\ell\le N}a_{k\ell}(x)v_1^kv_2^\ell,
\qquad
a_{k\ell}\in C^\infty(\overline{\Omega_+}),
\label{eq:main-integral}
\end{equation}
which is invariant under the magnetic flow between collisions and under specular reflection at \(\gamma\). Then \(\gamma\) is a circle. Equivalently, \(\Omega\) is a disk.
\end{theorem}

Here \(\Omega_+\) is the swept region appearing in the strong-field phase-space construction recalled in Section \ref{sec:setup}. The regularity in \eqref{eq:main-integral} is the smooth-up-to-the-boundary hypothesis used in \cite[Definition~2.2]{BialyMironovShalom2020}.

The proof is controlled by the two highest Fourier modes of the reflection law. Let \(m\ge1\) be the minimal true Fourier degree of the integral on the unit-velocity shell. Passing to Larmor-center coordinates produces a polynomial \(F(z,w)\) whose degree in each complex coordinate is at most \(m\). Along the physical boundary, its highest reflection mode has the form
\begin{equation}
P_m(\bar z)t^{2m}=(-1)^mQ_m(z),
\label{eq:intro-top}
\end{equation}
where \(t\) is the positively oriented unit tangent and \(P_m(w)=\overline{Q_m(\bar w)}\).

The proof then has two independent rigidity steps. First, logarithmic differentiation of \eqref{eq:intro-top}, together with the turning-number identity for a positively oriented smooth Jordan curve, yields
\begin{equation}
\deg Q_m=m,
\qquad
\operatorname{Zeros}(Q_m)\subset\Omega.
\label{eq:intro-index}
\end{equation}
Thus the highest coefficient has the largest degree allowed by the Fourier structure, and none of its roots lies on or outside the billiard table. This does not yet imply circularity: the \(m\) roots may still be distinct.

For the second step, polynomial integrability implies that the boundary is algebraic. The reflection identities therefore extend meromorphically to the normalization of its irreducible complexification. At a hypothetical point where both coordinate functions \(z\) and \(w\) have poles of orders \(a,b>0\), the highest reflection mode gives
\[
q=(-1)^mp\left(\frac ab\right)^m,
\]
whereas the next mode gives
\[
q=(-1)^{m-1}p\left(\frac ab\right)^{m-1}.
\]
Their comparison would require \(a=-b\), which is impossible. Hence the two coordinate functions cannot have a common pole. At a pole of one coordinate, the other is finite, and the highest reflection identity then forces that finite limiting coordinate to be a root of multiplicity \(m\). Consequently
\begin{equation}
Q_m(z)=\lambda(z-a)^m
\label{eq:intro-collapse}
\end{equation}
for some \(a\in\Omega\), where the location of \(a\) comes from \eqref{eq:intro-index}. Substituting \eqref{eq:intro-collapse} back into \eqref{eq:intro-top} shows that the boundary tangent makes a constant angle with the radial direction from \(a\). The radial distance has constant derivative, and closedness forces that derivative to vanish. Hence \(|z-a|\) is constant on \(\gamma\).

The theorem concerns polynomial integrability only. It does not assert that every totally integrable magnetic billiard possesses a polynomial first integral, and it does not address fields outside the strong-field regime \eqref{eq:strong-field}. No real-analyticity hypothesis on the original boundary is imposed: the algebraic strong-field theory and the regular inverse-offset relation provide the analytic structure required for the complex part of the proof.

The remainder of the paper is devoted to the proof. Section \ref{sec:setup} fixes the strong-field geometry and true Fourier degree. Sections \ref{sec:center} and \ref{sec:reflection} derive the Larmor-center polynomial and its two leading reflection identities. Section \ref{sec:bridge} passes to the compact normalization. Section \ref{sec:index} proves \eqref{eq:intro-index}, Section \ref{sec:poles} proves the root collapse \eqref{eq:intro-collapse}, and Section \ref{sec:mainproof} completes the circularity argument.

\section{Geometry and integrability in the strong-field regime}
\label{sec:setup}

Let \(\Omega\subset\mathbb{R}^{2}\) be a bounded domain with positively oriented \(C^\infty\) simple closed boundary
\[
\gamma=\partial\Omega.
\]
Let \(s\) be arclength, and write
\[
\gamma(s)=(x(s),y(s)).
\]
We identify \(\mathbb{R}^{2}\cong\mathbb{C}\) by
\[
z=x+\ii y.
\]
The positively oriented unit tangent is
\[
t(s)=z'(s)=e^{\ii\tau(s)},
\]
and the signed curvature is
\[
\kappa(s)=\tau'(s).
\]
The signed curvature may vanish or change sign. The global curvature identity used below is only the turning-number formula
\[
\int_\gamma \kappa\,ds=2\pi,
\]
valid for every positively oriented smooth Jordan curve.

Reversing the orientation of the plane if necessary, we may and shall assume that the magnetic field is positive. The case \(B<0\) is equivalent under orientation reversal. Let \(J\) be counterclockwise rotation through \(\pi/2\), identified in complex notation with multiplication by \(\ii\). A unit-speed magnetic trajectory then follows a positively oriented Larmor circle of radius
\[
r=\beta^{-1},
\qquad
\beta=|B|>0.
\]

Let \(r_0(\gamma)\) denote the maximal radius for which the normal exponential map of \(\gamma\) defines an embedded tubular neighborhood. We assume throughout
\begin{equation}
r<\frac{r_0(\gamma)}{2}.
\end{equation}
Under this hypothesis the parallel curves
\[
\gamma_{\pm r}(s)=\gamma(s)\pm rJ\dot\gamma(s)
\]
are smooth and embedded, and so are all parallel curves at distances up to \(2r\); compare \cite[\S2]{BialyMironovShalom2020}.

\subsection{The relevant phase-space region}

Following \cite{BialyMironovShalom2020}, let \(\Omega_+\) denote the annular region swept out by all Larmor arcs lying inside \(\Omega\) and intersecting \(\partial\Omega\). Equivalently, in the notation of that work, \(\Omega_+\) is the region bounded by \(\gamma\) and \(\gamma_{+2r}\). The precise orientation convention is irrelevant for the proof, provided it is fixed consistently.

\begin{definition}[Polynomial first integral]\label{def:polynomial-integral}
A function
\begin{equation}
\Phi(x,v)=\sum_{0\leq k+\ell\leq N}a_{k\ell}(x)v_1^kv_2^\ell,
\qquad
a_{k\ell}\in C^\infty(\overline{\Omega_+}),
\label{eq:polynomial-integral}
\end{equation}
defined on the unit tangent bundle of \(\Omega_+\), is called a polynomial first integral of the magnetic billiard if:
\begin{enumerate}
\item \(\Phi\) is constant along every magnetic-flow arc between collisions;
\item \(\Phi\) is invariant under specular reflection at every boundary point.
\end{enumerate}
\end{definition}

This is the strong-field definition used in \cite[Definition 2.2]{BialyMironovShalom2020}.

\subsection{True Fourier degree}

On the unit shell, write
\[
v=(\cos\varphi,\sin\varphi).
\]
The restriction of \(\Phi\) to unit velocities is a finite Fourier series
\begin{equation}
\Phi(x,\varphi)=\sum_{k=-m}^{m}a_k(x)e^{\ii k\varphi},
\qquad
a_{-k}=\overline{a_k}.
\label{eq:true-fourier}
\end{equation}

\begin{definition}[Minimal true Fourier degree]
The minimal integer \(m\geq0\) for which \eqref{eq:true-fourier} holds is called the \emph{true Fourier degree} of \(\Phi\) on the unit-velocity shell.
\end{definition}

A polynomial in \(v_1,v_2\) of nominal degree \(N\) may have smaller true Fourier degree because
\[
v_1^2+v_2^2=1
\]
on the unit shell. We always work with the minimal true degree \(m\), and we never replace the integral by a polynomial expression that artificially increases this degree.

\begin{lemma}[The zero-mode case is trivial]\label{lem:m-zero}
If the true Fourier degree is \(m=0\), then the first integral is constant on the connected phase-space region. Consequently every nonconstant polynomial first integral has \(m\geq1\).
\end{lemma}

\begin{proof}
If \(m=0\), then on the unit tangent bundle
\[
\Phi(x,v)=a_0(x)
\]
is independent of the velocity. Fix a boundary point \(x\in\gamma\). For every inward unit vector \(v\) at \(x\), invariance along the magnetic flow gives
\[
0=\frac{d}{dt}\Phi(g^t(x,v))\Big|_{t=0}
=
\nabla a_0(x)\cdot v.
\]
The inward unit vectors form an open semicircle, so the only vector orthogonal to all of them is the zero vector. Hence
\[
\nabla a_0(x)=0
\qquad
\text{for every }x\in\gamma.
\]
In particular, \(a_0|_\gamma\) is constant because \(\gamma\) is connected.

By definition, every point of the swept region \(\Omega_+\) lies on at least one Larmor arc contained in \(\Omega\) and intersecting \(\partial\Omega\). Along such an arc the function \(\Phi=a_0\) is constant, so its value at the interior point equals the constant boundary value. Therefore \(a_0\), and hence \(\Phi\), is constant on the whole relevant phase-space region.
\end{proof}

We do not replace the center-space polynomial by \( (F-c_+)(F-c_-) \) to normalize its values on the two parallel curves, because that operation can double the true Fourier degree.

\begin{example}[The disk]
For a disk centered at the origin, \eqref{eq:disk-integral} is a nonconstant polynomial integral. Its restriction to the unit shell has true Fourier degree one.
\end{example}

\section{From a velocity integral to a Larmor-center polynomial}
\label{sec:center}

The Larmor-center formulation turns the first integral into a function \(F\) on the center plane. The published strong-field theory makes \(F\) polynomial; the point needed here is the sharper bidegree bound in terms of the minimal Fourier degree.

\subsection{The center map}

Define the Larmor-center map
\begin{equation}
\mathcal{L}(x,v)=x+rJv.
\label{eq:center-map}
\end{equation}
By the Larmor-center construction of Bialy--Mironov--Shalom \cite[\S2.3]{BialyMironovShalom2020}, the first integral induces a well-defined function \(F\) on the annulus \(\Omega_r\) of Larmor centers such that
\begin{equation}
F\circ\mathcal{L}=\Phi.
\label{eq:F-compose}
\end{equation}

We first record the center-space polynomiality input that will be used below.

\begin{theorem}[Bialy--Mironov--Shalom]\label{thm:BMS-input}
Under the strong-field hypotheses above:
\begin{enumerate}
\item \(F\in C^\infty(\overline{\Omega_r})\);
\item if \(F\) restricts on every boundary-centered circle
\[
C_s=\{y:|y-\gamma(s)|=r\}
\]
to a polynomial of degree at most \(N\), then \(F\) is a polynomial in the center coordinates of total degree at most \(2N\);
\item \(F\) is constant on each of the parallel curves \(\gamma_{+r}\) and \(\gamma_{-r}\).
\end{enumerate}
\end{theorem}

\begin{proof}
These are Theorems 2.3 and 2.4 and Proposition 2.5 of \cite{BialyMironovShalom2020}.
\end{proof}

We now verify explicitly that Theorem \ref{thm:BMS-input} applies with the \emph{minimal true Fourier degree} \(N=m\), rather than merely with the nominal velocity degree of the original polynomial expression. Fix \(s\), write
\[
\gamma(s)=(x_s,y_s),
\]
and parametrize the corresponding boundary-centered Larmor circle by
\[
X-x_s=r\cos\theta,
\qquad
Y-y_s=r\sin\theta,
\qquad
\zeta=e^{\ii\theta}.
\]
On this circle the first integral has the form
\begin{equation}
F|_{C_s}
=
\sum_{k=-m}^{m}c_k(s)\zeta^k.
\label{eq:circle-fourier-ms}
\end{equation}
For \(k\geq1\),
\[
\zeta^k
=
r^{-k}\bigl((X-x_s)+\ii(Y-y_s)\bigr)^k,
\qquad
\zeta^{-k}
=
r^{-k}\bigl((X-x_s)-\ii(Y-y_s)\bigr)^k.
\]
Therefore the Cartesian polynomial
\begin{align}
R_s(X,Y)
={}&c_0(s)
+
\sum_{k=1}^{m}r^{-k}c_k(s)
\bigl((X-x_s)+\ii(Y-y_s)\bigr)^k
\notag\\
&+
\sum_{k=1}^{m}r^{-k}c_{-k}(s)
\bigl((X-x_s)-\ii(Y-y_s)\bigr)^k
\label{eq:explicit-circle-polynomial}
\end{align}
has total degree at most \(m\) and agrees with \(F\) on \(C_s\). The polynomial \(R_s\) may depend on \(s\), exactly as allowed in the hypothesis of the center-polynomial theorem. Consequently Theorem \ref{thm:BMS-input} applies with \(N=m\), and hence
\begin{equation}
\deg F\leq 2m.
\label{eq:F-totaldegree}
\end{equation}

Thus the center-space polynomial satisfies the total-degree bound \eqref{eq:F-totaldegree} with the minimal true Fourier degree itself.

\subsection{Complex coordinates}

Introduce independent complex variables
\[
Z=X+\ii Y,
\qquad
W=X-\ii Y.
\]
The real plane is the antiholomorphic diagonal
\[
W=\overline Z.
\]
Complexify \(F\) as
\begin{equation}
F(Z,W)=\sum_{a+b\leq2m}c_{ab}Z^aW^b.
\label{eq:F-complex}
\end{equation}
Since \(F\) is real on the real center plane,
\begin{equation}
c_{ab}=\overline{c_{ba}}.
\label{eq:Hermitian}
\end{equation}

The total-degree bound \eqref{eq:F-totaldegree} is not yet sharp enough for the top-mode argument. The next proposition is.

\begin{proposition}[Sharp bidegree reduction]\label{prop:bidegree}
The center polynomial has bidegree at most \((m,m)\):
\begin{equation}
\boxed{
F(Z,W)=\sum_{a=0}^{m}\sum_{b=0}^{m}c_{ab}Z^aW^b.
}
\label{eq:bidegree}
\end{equation}
\end{proposition}

\begin{proof}
Fix a boundary point
\[
z=z(s),\qquad w=\overline{z(s)}.
\]
On the corresponding boundary-centered Larmor circle write
\begin{equation}
Z=z+r\zeta,
\qquad
W=w+r\zeta^{-1},
\qquad
|\zeta|=1.
\label{eq:circle-param}
\end{equation}
Then
\[
F(z+r\zeta,w+r\zeta^{-1})
\]
has Laurent degree at most \(m\), because it equals the restriction of the original first integral to unit velocities at the boundary point \(z(s)\).

Let
\[
A=\max\{a:c_{ab}\neq0\text{ for some }b\}.
\]
Suppose \(A>m\). By maximality of \(A\), the coefficient of \(\zeta^A\) in
\[
F(z+r\zeta,w+r\zeta^{-1})
\]
is exactly
\[
r^A\sum_b c_{Ab}w^b.
\]
Since no harmonic \(\zeta^A\) can occur, we have
\[
\sum_b c_{Ab}w(s)^b=0
\]
for every \(s\). The left-hand side is a one-variable polynomial in \(w\) vanishing on the infinite set \(\{\overline{z(s)}\}\), hence it vanishes identically. This contradicts the definition of \(A\). Therefore \(a\leq m\) for every nonzero coefficient.

The identical argument applied to negative Laurent powers gives \(b\leq m\). Thus \eqref{eq:bidegree} holds.
\end{proof}

Write
\begin{equation}
F(Z,W)=\sum_{j=0}^{m}Z^jP_j(W),
\qquad
P_j(W)=\sum_{b=0}^{m}c_{jb}W^b.
\label{eq:row-polynomials}
\end{equation}
Define the algebraic conjugates
\begin{equation}
Q_j(Z)=\overline{P_j(\overline Z)}.
\label{eq:Qj}
\end{equation}
By \eqref{eq:Hermitian}, these are precisely the corresponding column polynomials.

\begin{lemma}[Nonvanishing of the top polynomial]\label{lem:top-nonzero}
If \(m\) is the minimal true Fourier degree, then
\[
P_m\not\equiv0,
\qquad
Q_m\not\equiv0.
\]
\end{lemma}

\begin{proof}
If \(P_m\equiv0\), then \eqref{eq:Hermitian} implies \(Q_m\equiv0\), so the \(m\)-th row and \(m\)-th column vanish. Proposition \ref{prop:bidegree} then gives bidegree at most \((m-1,m-1)\), and the restriction to every Larmor circle has Fourier degree at most \(m-1\), contradicting minimality.
\end{proof}

\section{Reflection symmetry and the two leading modes}
\label{sec:reflection}

The center polynomial still contains much more information than is needed.  Reflection acts on the circle parameter by a simple inversion involving the squared boundary tangent.  Comparing Laurent modes converts the billiard reflection law into algebraic identities.  Only the highest two modes will be needed for rigidity, and this section computes them explicitly.

Fix a boundary point \(z\in\gamma\), put \(w=\overline z\), and expand
\begin{equation}
F(z+r\zeta,w+r\zeta^{-1})
=
\sum_{k=-m}^{m}C_k(z,w)\zeta^k.
\label{eq:Laurent-expansion}
\end{equation}
For \(k>0\), write
\begin{equation}
C_k=r^kA_k(z,w),
\qquad
C_{-k}=r^kB_k(z,w).
\label{eq:AkBk}
\end{equation}
On the real boundary,
\[
B_k(z,\overline z)=\overline{A_k(z,\overline z)}.
\]

\subsection{The top two coefficients}

\begin{lemma}[Top two positive modes]\label{lem:top-two}
The highest two positive coefficients are
\begin{align}
A_m&=P_m(w),
\label{eq:Am}\\
A_{m-1}
&=
P_{m-1}(w)+mzP_m(w)+r^2P_m'(w).
\label{eq:Am1}
\end{align}
Similarly,
\begin{align}
B_m&=Q_m(z),
\label{eq:Bm}\\
B_{m-1}
&=
Q_{m-1}(z)+mwQ_m(z)+r^2Q_m'(z).
\label{eq:Bm1}
\end{align}
\end{lemma}

\begin{proof}
Using \eqref{eq:row-polynomials} and \eqref{eq:circle-param},
\[
F(z+r\zeta,w+r\zeta^{-1})
=
\sum_{j=0}^{m}(z+r\zeta)^jP_j(w+r\zeta^{-1}).
\]

For \(\zeta^m\), the only possibility is to take all \(m\) positive powers from \((z+r\zeta)^m\) and no negative power from \(P_m\), giving
\[
C_m=r^mP_m(w).
\]

For \(\zeta^{m-1}\), exactly three contributions occur:
\[
r^{m-1}P_{m-1}(w),
\qquad
mzr^{m-1}P_m(w),
\qquad
r^{m+1}P_m'(w).
\]
After dividing by \(r^{m-1}\), this gives \eqref{eq:Am1}. The negative-mode formulas follow by algebraic conjugation.
\end{proof}

\subsection{Reflection in the Larmor-center variable}

Represent the unit velocity as a complex number \(v\) with \(|v|=1\). Define
\[
\zeta=\ii v.
\]
Then the Larmor center is
\[
Z=z+r\zeta.
\]

Let \(t\) be the unit tangent at the reflection point. Specular reflection across the tangent line sends
\[
v_-\longmapsto v_+=t^2\overline{v_-}.
\]
Since \(\zeta_-=\ii v_-\), we have
\[
\overline{v_-}=\frac{\ii}{\zeta_-},
\]
and therefore
\begin{equation}
\boxed{
\zeta_+
=
-\frac{t^2}{\zeta_-}.
}
\label{eq:zeta-reflection}
\end{equation}

Set
\begin{equation}
u=t^2.
\label{eq:u-def}
\end{equation}
Reflection invariance gives
\[
F(z+r\zeta,w+r\zeta^{-1})
=
F\left(z-r\frac{u}{\zeta},
w-r\frac{\zeta}{u}\right).
\]
Comparing Laurent coefficients yields the universal reflection law.

\begin{proposition}[Fourier reflection identities]\label{prop:reflection-identities}
For every \(1\leq k\leq m\),
\begin{equation}
\boxed{
A_k u^k=(-1)^kB_k.
}
\label{eq:reflection-mode}
\end{equation}
In particular,
\begin{equation}
\boxed{
P_m(w)u^m=(-1)^mQ_m(z)
}
\label{eq:top-reflection}
\end{equation}
and
\begin{equation}
\boxed{
\begin{aligned}
&\bigl(P_{m-1}(w)+mzP_m(w)+r^2P_m'(w)\bigr)u^{m-1}\\
&\qquad=
(-1)^{m-1}
\bigl(Q_{m-1}(z)+mwQ_m(z)+r^2Q_m'(z)\bigr).
\end{aligned}
}
\label{eq:next-reflection}
\end{equation}
\end{proposition}

\begin{proof}
For a fixed boundary point, physical reflection invariance gives equality of the two Laurent expressions for the open arc of \(|\zeta|=1\) corresponding to incoming unit velocities. Both sides are Laurent polynomials in \(\zeta\), hence holomorphic on \(\mathbb{C}^{\times}\). Since they agree on an arc with an accumulation point in \(\mathbb{C}^{\times}\), the identity theorem extends the equality to all \(\zeta\in\mathbb{C}^{\times}\). We may therefore compare Laurent coefficients.

Substitute \(\zeta\mapsto-u/\zeta\) in the Laurent polynomial \eqref{eq:Laurent-expansion}. The coefficient of \(\zeta^k\) on the reflected side is
\[
(-1)^ku^{-k}C_{-k}.
\]
Equality with \(C_k\) gives
\[
C_ku^k=(-1)^kC_{-k}.
\]
Using \eqref{eq:AkBk} proves \eqref{eq:reflection-mode}; \eqref{eq:top-reflection} and \eqref{eq:next-reflection} follow from Lemma \ref{lem:top-two}.
\end{proof}

\section{From the real boundary to a compact algebraic curve}
\label{sec:bridge}

To compare the reflection identities at infinity, we pass from the physical boundary to the normalization of an irreducible complex algebraic curve. Polynomial integrability first upgrades the smooth boundary to a real-analytic algebraic curve.

\subsection{Automatic analyticity from a nonsingular algebraic parallel curve}

\begin{proposition}[Automatic analyticity]\label{prop:auto-analyticity}
Under the hypotheses of Theorem \ref{thm:main}, the smooth boundary \(\gamma\) is real analytic.
\end{proposition}

\begin{proof}
By Theorem 1.1 of \cite{BialyMironovShalom2020}, polynomial integrability in the strong-field regime implies that each parallel curve \(\gamma_{\pm r}\) is a real oval of a nonsingular affine algebraic curve.  In particular,
\[
\Gamma:=\gamma_{+r}
\]
is a regular real-analytic embedded closed curve.

Let \(s\) be arclength on the original smooth boundary, let
\[
t(s)=\gamma'(s),
\qquad
n(s)=Jt(s),
\]
and use the convention
\[
\Gamma(s)=\gamma(s)+r n(s).
\]
With signed curvature \(\kappa\) defined by
\[
t'=\kappa n,
\qquad
n'=-\kappa t,
\]
one has
\begin{equation}
\Gamma'(s)=\bigl(1-r\kappa(s)\bigr)t(s).
\label{eq:parallel-tangent-factor}
\end{equation}
Because \(r<r_0(\gamma)\), the normal map
\[
E:\gamma\times(-r_0(\gamma),r_0(\gamma))\longrightarrow\mathbb R^2,
\qquad
E(s,\rho)=\gamma(s)+\rho n(s),
\]
is an embedding.  Hence its tangential differential never vanishes, so
\[
1-\rho\kappa(s)\neq0
\qquad
\text{for }0\leq\rho\leq r.
\]
Since this factor equals \(1\) at \(\rho=0\), it remains positive throughout the interval.  Therefore \eqref{eq:parallel-tangent-factor} shows that \(\Gamma\) has the same oriented unit tangent and the same chosen unit normal as \(\gamma\):
\[
t_\Gamma(\Gamma(s))=t(s),
\qquad
n_\Gamma(\Gamma(s))=n(s).
\]
Consequently the original boundary is recovered pointwise by the inverse offset formula
\begin{equation}
\boxed{
\gamma(s)=\Gamma(s)-r n_\Gamma(\Gamma(s)).
}
\label{eq:inverse-offset}
\end{equation}

Now choose any local real-analytic parameter \(\sigma\) on the nonsingular algebraic oval \(\Gamma\).  Its unit tangent and unit normal are real analytic in \(\sigma\), so
\[
\widetilde\gamma(\sigma)
:=
\Gamma(\sigma)-r n_\Gamma(\sigma)
\]
is a real-analytic parametrized curve.  By \eqref{eq:inverse-offset} and uniqueness in the tubular neighborhood, its image is exactly the original boundary \(\gamma\).  Thus \(\gamma\) is real analytic.
\end{proof}

\subsection{Algebraicity of the original boundary}

The same strong-field algebraic framework also gives algebraicity of the original boundary from the algebraic parallel curves.

\begin{proposition}[Algebraicity of \(\gamma\)]\label{prop:gamma-algebraic}
Under the hypotheses of Theorem \ref{thm:main}, the original boundary \(\gamma\) is contained in a real affine algebraic curve.
\end{proposition}

\begin{proof}
By \cite[\S2.3]{BialyMironovShalom2020}, the algebraicity of the regular parallel curves \(\gamma_{\pm r}\) implies that the original boundary \(\gamma\) is algebraic. Hence a nonzero real polynomial vanishes on \(\gamma\).
\end{proof}

\subsection{A single irreducible complexification}

Change to independent complex coordinates
\[
z=x+\ii y,
\qquad
w=x-\ii y.
\]
Let \(H(z,w)\in\mathbb{C}[z,w]\) be a polynomial vanishing on
\[
\Gamma_{\mathbb{R}}
=
\{(z,\overline z):z\in\gamma\}.
\]
Factor
\[
H=f_1\cdots f_N
\]
into irreducibles over \(\mathbb{C}\).

\begin{lemma}[One irreducible component contains the full oval]\label{lem:one-component}
There exists an irreducible factor \(f\in\mathbb{C}[z,w]\) such that
\[
f(z(s),\overline{z(s)})\equiv0
\]
for all \(s\).
\end{lemma}

\begin{proof}
For each factor \(f_j\), the function
\[
s\longmapsto f_j(z(s),\overline{z(s)})
\]
is real analytic. Their product vanishes identically on the parameter circle. If no factor vanished identically, each zero set would be discrete and therefore finite on the compact parameter circle. A finite union of finite sets cannot cover the whole circle. Thus at least one factor vanishes identically.
\end{proof}

Fix such an irreducible polynomial \(f\), and write
\[
C=\{(z,w)\in\mathbb{C}^{2}:f(z,w)=0\}.
\]
Let \(\overline C\subset\mathbb{CP}^{2}\) be the projective closure, and let
\[
\nu:\widehat C\longrightarrow\overline C
\]
be its normalization. The compact Riemann surface \(\widehat C\) is the natural stage for the valuation argument. Since an irreducible plane curve has only finitely many affine singular points, the physical oval contains a nonempty open arc in the regular locus of \(C\). That arc has a unique local lift to \(\widehat C\), and all real-boundary identities below are first read on this lifted regular arc.

\subsection{The intrinsic meromorphic tangent variable}

Along the physical real oval,
\[
w=\overline z,
\qquad
\frac{dz}{ds}=t,
\qquad
\frac{dw}{ds}=\overline t=\frac1t.
\]
Therefore
\begin{equation}
\frac{dz}{dw}=t^2=u.
\label{eq:u-differential}
\end{equation}

The intrinsic definition of the tangent variable is
\begin{equation}
\boxed{
u:=\frac{dz}{dw}.
}
\label{eq:u-primary}
\end{equation}
This is the primary definition used throughout the global argument. The affine formula involving partial derivatives is only a local representation of the same meromorphic function.

To see this directly, choose any local parameter \(\xi\) on \(\widehat C\). Then
\[
dz=z_\xi(\xi)\,d\xi,
\qquad
dw=w_\xi(\xi)\,d\xi,
\]
with \(z_\xi\) and \(w_\xi\) meromorphic. Since \(w\) is nonconstant on \(\widehat C\), \(w_\xi\) is not identically zero, and
\[
u=\frac{z_\xi}{w_\xi}
\]
is meromorphic after canceling common powers in the local parameter. Thus ramification points, zeros of both differentials, and singularities of the affine plane model create at most zeros or poles of \(u\); they do not create any ambiguity on the normalization.

On the regular affine locus of \(C\), the following identity holds meromorphically. Differentiating
\[
f(z,w)=0
\]
gives
\[
f_z\,dz+f_w\,dw=0,
\]
and therefore
\begin{equation}
\boxed{
u=\frac{dz}{dw}=-\frac{f_w}{f_z}.
}
\label{eq:u-f}
\end{equation}

\begin{proposition}[Meromorphic continuation of the tangent variable]\label{prop:u-meromorphic}
The intrinsically defined function
\[
u:=\frac{dz}{dw}
\]
is meromorphic on \(\widehat C\), agrees with \(t^2\) on the physical oval, and is represented on the regular affine locus by \(-f_w/f_z\).
\end{proposition}

\begin{proof}
The coordinate functions \(z,w\) pull back to nonconstant meromorphic functions on \(\widehat C\). In a local parameter \(\xi\), the quotient of their differentials is the quotient \(z_\xi/w_\xi\), hence a meromorphic function. Along the physical oval,
\[
\frac{dz}{ds}=t,
\qquad
\frac{dw}{ds}=\overline t=t^{-1},
\]
so
\[
\frac{dz}{dw}=t^2.
\]
On the regular affine locus, the differential identity \(f_zdz+f_wdw=0\) gives \eqref{eq:u-f}.
\end{proof}

\subsection{Global extension of the reflection identities}

\begin{proposition}[Global bridge]\label{prop:global-bridge}
Every reflection identity \eqref{eq:reflection-mode} extends as a meromorphic identity on \(\widehat C\). In particular, \eqref{eq:top-reflection} and \eqref{eq:next-reflection} hold identically on \(\widehat C\).
\end{proposition}

\begin{proof}
For each \(k\), define on \(\widehat C\)
\[
R_k
=
A_k(z,w)u^k-(-1)^kB_k(z,w).
\]
Because \(A_k\) and \(B_k\) are polynomials and \(u\) is meromorphic, \(R_k\) is meromorphic. It vanishes on the lifted nonempty regular arc of the physical real oval, where \(u=t^2\) and the original reflection law holds. The zero set has an accumulation point on the Riemann surface, so the identity theorem gives
\[
R_k\equiv0
\]
on \(\widehat C\).
\end{proof}

\section{Boundary winding determines the degree and location of the top-mode roots}
\label{sec:index}

Set
\[
Q(z)=Q_m(z).
\]
On the physical boundary, \eqref{eq:top-reflection} becomes
\begin{equation}
\boxed{
\overline{Q(z)}\,t^{2m}=(-1)^mQ(z).
}
\label{eq:top-real}
\end{equation}

\subsection{The logarithmic-derivative identity}

\begin{lemma}\label{lem:log-derivative}
At every boundary point where \(Q(z)\neq0\),
\begin{equation}
\boxed{
\operatorname{Im}
\left(
\frac{Q'(z)}{Q(z)}t
\right)
=
m\kappa.
}
\label{eq:logderiv}
\end{equation}
\end{lemma}

\begin{proof}
From \eqref{eq:top-real},
\[
t^{2m}=(-1)^m\frac{Q(z)}{\overline{Q(z)}}.
\]
Differentiate logarithmically with respect to arclength \(s\). Since
\[
t=e^{\ii\tau},
\qquad
\tau'=\kappa,
\]
the left side contributes \(2m\ii\kappa\). The right side contributes
\[
\frac{Q'(z)}{Q(z)}t
-
\overline{\frac{Q'(z)}{Q(z)}t}.
\]
Dividing by \(2\ii\) gives \eqref{eq:logderiv}.
\end{proof}

\subsection{Boundary roots and principal values}

The dangerous case is a root of \(Q\) lying on \(\gamma\). We treat it explicitly.

\begin{lemma}[Local form at a boundary root]\label{lem:boundary-root-local}
Suppose \(z_0=z(s_0)\in\gamma\) is a root of \(Q\) of multiplicity \(\nu\geq1\). Then
\begin{equation}
\frac{Q'(z(s))}{Q(z(s))}t(s)
=
\frac{\nu}{s-s_0}+O(1).
\label{eq:boundary-root-local}
\end{equation}
In particular, the singular principal part is real.
\end{lemma}

\begin{proof}
Write
\[
Q(z)=(z-z_0)^\nu R(z),
\qquad
R(z_0)\neq0.
\]
Since \(s\) is arclength,
\[
z(s)-z_0=t(s_0)(s-s_0)+O((s-s_0)^2).
\]
Hence
\[
\frac{Q'}Q\,z'
=
\nu\frac{z'}{z-z_0}
+
\frac{R'}R\,z'
=
\frac{\nu}{s-s_0}+O(1).
\]
\end{proof}

Thus the imaginary part in \eqref{eq:logderiv} has no nonintegrable singularity at a boundary zero.

\subsection{The half-index lemma}

\begin{lemma}[Principal-value half residue]\label{lem:half-residue}
Let \(\gamma\) be a positively oriented \(C^2\) Jordan curve and let \(\alpha\in\gamma\). Parametrize \(\gamma\) by arclength near \(\alpha=z(s_0)\), and define the symmetric arclength principal value by
\[
\PV\oint_\gamma\frac{dz}{z-\alpha}
:=
\lim_{\varepsilon\downarrow0}
\int_{\gamma\setminus z((s_0-\varepsilon,s_0+\varepsilon))}
\frac{dz}{z-\alpha},
\]
whenever the limit exists. Then
\begin{equation}
\boxed{
\PV\oint_\gamma\frac{dz}{z-\alpha}=\ii\pi.
}
\label{eq:half-residue}
\end{equation}
\end{lemma}

\begin{proof}
We give a contour-indentation proof, which fixes both the factor \(1/2\) and the sign. Translate and rotate so that \(\alpha=0\) and the positively oriented tangent at \(\alpha\) points along the positive real axis. The interior of the Jordan domain then lies locally on the left side of the oriented tangent.

For sufficiently small \(\rho>0\), the circle \(|z|=\rho\) meets \(\gamma\) in two points near the two tangent directions \(\pm1\). Remove the short boundary arc passing through \(0\), and close the remaining positively oriented boundary by the circular arc \(A_\rho\subset\{|z|=\rho\}\) lying inside the domain. The orientation induced on \(A_\rho\) is clockwise. The resulting indented contour \(\Gamma_\rho\) contains no pole of \(1/z\), so Cauchy's theorem gives
\begin{equation}
0
=
\int_{\Gamma_\rho}\frac{dz}{z}
=
\int_{\gamma\setminus U_\rho}\frac{dz}{z}
+
\int_{A_\rho}\frac{dz}{z},
\label{eq:indent-zero}
\end{equation}
where \(U_\rho\) is the deleted short boundary arc.

Write \(z=\rho e^{\ii\theta}\) on \(A_\rho\). Because \(\gamma\) is \(C^2\), the curve is a regular graph over its tangent line in a sufficiently small neighborhood of \(\alpha\); consequently the small circle meets the local boundary in exactly two points, whose directions converge to the opposite tangent directions. Thus the clockwise angular variation along \(A_\rho\) is
\[
-\pi+o(1).
\]
Therefore
\[
\int_{A_\rho}\frac{dz}{z}
=
\ii\,\Delta\theta
=
-\ii\pi+o(1).
\]
Equation \eqref{eq:indent-zero} now yields
\[
\int_{\gamma\setminus U_\rho}\frac{dz}{z}
=
\ii\pi+o(1).
\]
The \(C^2\) expansion
\[
z(s)-\alpha=t(s_0)(s-s_0)+O((s-s_0)^2)
\]
shows that replacing the small-circle cutoff by the symmetric arclength cutoff changes the truncated integral by \(o(1)\): the two endpoint logarithmic moduli have the same leading behavior and the local angular discrepancy tends to zero. Passing to the limit gives
\[
\PV\oint_\gamma\frac{dz}{z-\alpha}=\ii\pi.
\]
\end{proof}

\begin{proposition}[Half-index formula]\label{prop:half-index}
Let \(Q\) be a nonzero complex polynomial. Let
\[
N_{\mathrm{in}},
\qquad
N_{\partial},
\qquad
N_{\mathrm{out}}
\]
be the numbers of its roots inside \(\gamma\), on \(\gamma\), and outside \(\gamma\), counted with multiplicity. If several roots lie on \(\gamma\), the principal value is understood by deleting symmetric arclength intervals around all boundary roots and then letting all deletion radii tend to zero. Then
\begin{equation}
\boxed{
\PV\operatorname{Im}
\oint_\gamma\frac{Q'(z)}{Q(z)}\,dz
=
2\pi N_{\mathrm{in}}+\pi N_{\partial}.
}
\label{eq:half-index}
\end{equation}
\end{proposition}

\begin{proof}
Factor
\[
Q(z)=c\prod_j(z-\alpha_j)^{\nu_j}.
\]
Then
\[
\frac{Q'}Q=\sum_j\frac{\nu_j}{z-\alpha_j}.
\]
An interior root contributes \(2\pi\ii\nu_j\), an exterior root contributes \(0\), and a boundary root contributes \(\pi\ii\nu_j\) by Lemma \ref{lem:half-residue}. Taking imaginary parts gives \eqref{eq:half-index}.
\end{proof}

\subsection{The universal index identity}

\begin{proposition}[Top-mode index rigidity]\label{prop:index-rigidity}
The top polynomial \(Q_m\) has degree exactly \(m\), and all \(m\) roots, counted with multiplicity, lie strictly inside \(\Omega\). Equivalently,
\begin{equation}
\boxed{
\deg Q_m=m,
\qquad
N_{\mathrm{in}}=m,
\qquad
N_{\partial}=N_{\mathrm{out}}=0.
}
\label{eq:index-conclusion}
\end{equation}
\end{proposition}

\begin{proof}
By Lemma \ref{lem:log-derivative} and Lemma \ref{lem:boundary-root-local},
\[
\PV\operatorname{Im}
\oint_\gamma\frac{Q_m'}{Q_m}\,dz
=
m\int_\gamma\kappa\,ds.
\]
Since \(\gamma\) is a positively oriented smooth Jordan curve,
\[
\int_\gamma\kappa\,ds=2\pi.
\]
Thus
\[
\PV\operatorname{Im}
\oint_\gamma\frac{Q_m'}{Q_m}\,dz
=
2\pi m.
\]
By Proposition \ref{prop:half-index},
\begin{equation}
2N_{\mathrm{in}}+N_{\partial}=2m.
\label{eq:index-identity}
\end{equation}

Set
\[
d=\deg Q_m.
\]
By Proposition \ref{prop:bidegree}, \(d\leq m\). On the other hand,
\[
2d
=
2N_{\mathrm{in}}+2N_{\partial}+2N_{\mathrm{out}}
\geq
2N_{\mathrm{in}}+N_{\partial}
=
2m.
\]
Hence \(d\geq m\), so \(d=m\). Equality further gives
\[
0=2d-2m=N_{\partial}+2N_{\mathrm{out}},
\]
and therefore
\[
N_{\partial}=N_{\mathrm{out}}=0,
\qquad
N_{\mathrm{in}}=m.
\]
\end{proof}

\section{Behavior at infinity forces the top roots to coalesce}
\label{sec:poles}

We now work on the compact normalization \(\widehat C\). By Proposition \ref{prop:global-bridge}, the identities
\begin{equation}
P_m(w)u^m=(-1)^mQ_m(z)
\label{eq:pole-top}
\end{equation}
and
\begin{equation}
\begin{aligned}
&\bigl(P_{m-1}(w)+mzP_m(w)+r^2P_m'(w)\bigr)u^{m-1}\\
&\qquad=
(-1)^{m-1}
\bigl(Q_{m-1}(z)+mwQ_m(z)+r^2Q_m'(z)\bigr)
\end{aligned}
\label{eq:pole-next}
\end{equation}
hold meromorphically on \(\widehat C\), where
\[
u=\frac{dz}{dw}.
\]

By Proposition \ref{prop:index-rigidity},
\[
\deg P_m=\deg Q_m=m.
\]
Write
\begin{equation}
P_m(w)=pw^m+\cdots,
\qquad
Q_m(z)=qz^m+\cdots,
\qquad
p,q\neq0.
\label{eq:leading-pq}
\end{equation}

Every point of \(\widehat C\) lying over the line at infinity falls into exactly one of three cases:
\begin{enumerate}
\item both \(z\) and \(w\) have poles;
\item \(z\) has a pole and \(w\) is finite;
\item \(w\) has a pole and \(z\) is finite.
\end{enumerate}
A point lying over the line at infinity cannot have both affine coordinate functions finite, so these three cases exhaust the points relevant to the valuation argument.

\subsection{Simultaneous poles are impossible}

\begin{proposition}[Simultaneous-pole obstruction]\label{prop:no-simultaneous}
Assume \(m\geq2\). There is no point of \(\widehat C\) at which both \(z\) and \(w\) have poles.
\end{proposition}

\begin{proof}
Suppose such a point exists. Choose a local parameter \(\xi\) vanishing there and write
\begin{equation}
z=A\xi^{-a}(1+O(\xi)),
\qquad
w=B\xi^{-b}(1+O(\xi)),
\label{eq:sim-poles}
\end{equation}
where
\[
a,b\in\mathbb{N}_{>0},
\qquad
A,B\neq0.
\]
Since
\[
u=\frac{dz}{dw},
\]
we obtain
\begin{equation}
u=
\frac{aA}{bB}\xi^{b-a}(1+O(\xi)).
\label{eq:u-sim}
\end{equation}

Substitute \eqref{eq:sim-poles} and \eqref{eq:u-sim} into the top equation \eqref{eq:pole-top}. Using \eqref{eq:leading-pq},
\[
pB^m\left(\frac{aA}{bB}\right)^m\xi^{-ma}
=
(-1)^mqA^m\xi^{-ma}
+O(\xi^{-ma+1}).
\]
Thus
\begin{equation}
\boxed{
q=(-1)^mp\left(\frac ab\right)^m.
}
\label{eq:q-top}
\end{equation}

Now examine the next-to-top coefficient
\[
A_{m-1}
=
P_{m-1}(w)+mzP_m(w)+r^2P_m'(w).
\]
The three valuations satisfy
\begin{align*}
\ord_\xi P_{m-1}(w)&\geq-mb,\\
\ord_\xi\bigl(mzP_m(w)\bigr)&=-a-mb,\\
\ord_\xi P_m'(w)&=-(m-1)b.
\end{align*}
Because \(a,b>0\),
\[
-a-mb<-mb,
\qquad
-a-mb<-(m-1)b.
\]
Therefore \(mzP_m(w)\) is the unique dominant term:
\begin{equation}
A_{m-1}
=
mpAB^m\xi^{-a-mb}(1+O(\xi)).
\label{eq:A-next-dom}
\end{equation}
Similarly,
\begin{equation}
B_{m-1}
=
mqBA^m\xi^{-b-ma}(1+O(\xi)).
\label{eq:B-next-dom}
\end{equation}

Multiplying \eqref{eq:A-next-dom} by
\[
u^{m-1}
=
\left(\frac{aA}{bB}\right)^{m-1}
\xi^{(m-1)(b-a)}(1+O(\xi))
\]
and comparing with \eqref{eq:B-next-dom} in \eqref{eq:pole-next} gives
\begin{equation}
\boxed{
q=(-1)^{m-1}p\left(\frac ab\right)^{m-1}.
}
\label{eq:q-next}
\end{equation}

Comparing \eqref{eq:q-top} and \eqref{eq:q-next} and canceling nonzero factors yields
\[
-\frac ab=1,
\]
that is,
\[
a=-b,
\]
impossible because \(a,b>0\).
\end{proof}

\subsection{A one-sided pole forces complete root coalescence}

\begin{proposition}[One-sided pole lemma]\label{prop:one-sided}
Suppose \(z\) has a pole at a point of \(\widehat C\) while \(w\to\beta\in\mathbb{C}\). Then \(\beta\) is a root of multiplicity \(m\) of \(P_m\). Symmetrically, if \(w\) has a pole while \(z\to\alpha\in\mathbb{C}\), then \(\alpha\) is a root of multiplicity \(m\) of \(Q_m\).
\end{proposition}

\begin{proof}
Write
\begin{equation}
z=A\xi^{-a}(1+O(\xi)),
\qquad
w=\beta+B\xi^b(1+O(\xi)),
\label{eq:one-sided-exp}
\end{equation}
with
\[
a,b>0,
\qquad
A,B\neq0.
\]
The positive integer \(b\) exists because \(w\) is a nonconstant meromorphic function: if it were locally constant, it would be globally constant by analytic continuation, impossible on the physical oval where \(w=\overline z\) and \(z\) is nonconstant.

From
\[
u=\frac{dz}{dw}
\]
we obtain
\begin{equation}
u=
-\frac{aA}{bB}\xi^{-a-b}(1+O(\xi)).
\label{eq:u-one-sided}
\end{equation}

Suppose \(\beta\) is a root of \(P_m\) of multiplicity \(n\), where \(0\leq n\leq m\):
\[
P_m(w)=c(w-\beta)^n+\cdots,
\qquad
c\neq0.
\]
Then
\[
P_m(w)u^m
\]
has valuation
\[
nb-ma-mb.
\]
Because \(\deg Q_m=m\), the right side of \eqref{eq:pole-top} has valuation
\[
-ma.
\]
Equality of the two nonzero meromorphic functions therefore forces
\[
nb-ma-mb=-ma.
\]
Hence
\[
(n-m)b=0.
\]
Since \(b>0\),
\[
n=m.
\]
The symmetric statement follows by interchanging \(z\) and \(w\).
\end{proof}

\begin{corollary}[Root collapse]\label{cor:root-collapse}
For every \(m\geq2\), there exist \(\lambda\neq0\) and \(a\in\Omega\) such that
\begin{equation}
\boxed{
Q_m(z)=\lambda(z-a)^m.
}
\label{eq:root-collapse}
\end{equation}
\end{corollary}

\begin{proof}
The meromorphic function \(z\) is nonconstant on the compact Riemann surface \(\widehat C\), so it has a pole. By Proposition \ref{prop:no-simultaneous}, \(w\) is finite at that point. Proposition \ref{prop:one-sided} therefore gives
\[
P_m(w)=\overline\lambda(w-\overline a)^m
\]
for some \(a\in\mathbb{C}\). Hence
\[
Q_m(z)=\lambda(z-a)^m.
\]
By Proposition \ref{prop:index-rigidity}, every root of \(Q_m\) lies strictly inside \(\Omega\), so \(a\in\Omega\).
\end{proof}

\section{From root collapse to circularity}
\label{sec:mainproof}

Proposition \ref{prop:index-rigidity} places all top-mode roots inside \(\Omega\), while Corollary \ref{cor:root-collapse} shows that for \(m\ge2\) they coincide. For \(m=1\), full degree already gives a single interior root. It remains to convert this algebraic conclusion into circularity.

\begin{lemma}[An \(m\)-fold top root forces a circle]\label{lem:mf-root-circle}
Suppose for some \(m\geq1\),
\[
Q_m(z)=\lambda(z-a)^m,
\qquad
\lambda\neq0,
\qquad
a\in\Omega.
\]
Then \(\gamma\) is a circle centered at \(a\).
\end{lemma}

\begin{proof}
On the real boundary,
\[
P_m(\overline z)=\overline{\lambda}(\overline z-\overline a)^m.
\]
The top reflection equation \eqref{eq:top-real} becomes
\begin{equation}
\overline\lambda(\overline z-\overline a)^m t^{2m}
=
(-1)^m\lambda(z-a)^m.
\label{eq:mfold-top}
\end{equation}
Because \(a\in\Omega\), the boundary never passes through \(a\). Define the continuous unit complex function
\[
\eta(s)
:=
t(s)\,
\frac{\overline{z(s)-a}}{|z(s)-a|}.
\]
Equation \eqref{eq:mfold-top} implies
\[
\eta(s)^{2m}
=
(-1)^m\frac{\lambda}{\overline\lambda}
=:\sigma,
\qquad
|\sigma|=1.
\]
Thus the image of the connected curve under \(\eta\) is contained in the finite set of \(2m\)-th roots of \(\sigma\). By continuity and connectedness,
\[
\eta(s)\equiv e^{\ii\delta}
\]
for some constant \(\delta\in\mathbb{R}\).

Set
\[
\rho(s)=|z(s)-a|.
\]
Then
\[
\frac{d\rho}{ds}
=
\operatorname{Re}
\left(
t(s)\frac{\overline{z(s)-a}}{|z(s)-a|}
\right)
=
\operatorname{Re}\eta(s)
=
\cos\delta,
\]
so \(d\rho/ds\) is constant. Closedness gives
\[
0=\rho(L)-\rho(0)=L\cos\delta,
\]
where \(L\) is the perimeter of \(\gamma\). Hence
\[
\cos\delta=0,
\]
so \(\rho\) is constant. Therefore
\[
|z-a|=R
\]
on \(\gamma\), and \(\gamma\) is a circle centered at \(a\).
\end{proof}

\begin{proof}[Proof of Theorem \ref{thm:main}]
Let \(m\) be the minimal true Fourier degree of the nonconstant polynomial first integral. By Lemma \ref{lem:m-zero}, one has \(m\geq1\).

If \(m=1\), Proposition \ref{prop:index-rigidity} gives
\[
\deg Q_1=1
\]
and places its unique root \(a\) inside \(\Omega\). Thus
\[
Q_1(z)=\lambda(z-a),
\]
and Lemma \ref{lem:mf-root-circle} gives circularity.

Assume \(m\geq2\). Corollary \ref{cor:root-collapse} gives
\[
Q_m(z)=\lambda(z-a)^m
\]
with \(a\in\Omega\). Lemma \ref{lem:mf-root-circle} again implies that \(\gamma\) is a circle. Therefore \(\Omega\) is a disk.
\end{proof}

\section*{Acknowledgments}

The author thanks the authors of the foundational works on algebraic and strong-field magnetic billiards, whose center-space framework makes the present argument possible. The author also thanks colleagues and collaborators whose broader discussions on exact integrability and rigidity motivated the search for a degree-independent obstruction.

\section*{Generative AI disclosure}

OpenAI's ChatGPT was used as an auxiliary research-assistance tool during the development and preparation of this work. The original research questions and initial scientific direction were conceived and initiated by the author; subsequent interaction with the model was used to explore additional directions, test candidate arguments, assist with symbolic and algebraic organization, perform adversarial consistency checks, identify potentially relevant literature, and improve the clarity and presentation of the manuscript.

All AI-generated outputs were treated as provisional and non-authoritative. Every theorem, proof, derivation, equation, citation, and scientific conclusion appearing in the final manuscript was independently examined, verified, and revised by the author. The identification and interpretation of the principal results, assessment of novelty and scientific significance, and final responsibility for the content of the manuscript rest solely with the author.

\end{document}